\documentclass[aps,pra,twocolumn,showpacs,superscriptaddress,reprint]{revtex4-1}
\usepackage{amsmath}
\usepackage{amsthm}
\usepackage{environ}
\usepackage{tabularx}

\newtheorem{corollary}{Corollary}
\newtheorem{lemma}{Lemma}
\newtheorem{theorem}{Theorem}

\def\|#1>{\left|#1\right\rangle}
\def\<#1|{\left\langle#1\right|}

\usepackage{tikz}
\usetikzlibrary{backgrounds,calc,decorations.pathreplacing,external}

\newcommand{\startqcirc}{%
  \begin{tikzpicture}[x=\qcircx,y=-\qcircy,scale=\qcircscale,baseline=(c)]
  \setcounter{qcirccol}{0}
  \setcounter{qcircrow}{0}
  \global\qcircadjl=0pt
  \global\qcircadjr=0pt
  \global\let\oldcr=\cr
  \def\>{\checkforemp\addtocounter{qcirccol}{1}}
  \def\cr{\checkforemp\setcounter{qcirccol}{0}\addtocounter{qcircrow}{1}}
}

\newif\ifusepdf \usepdffalse
\newcounter{currentfigure}
\NewEnviron{qcirc}{%
  \ifusepdf
  \vcenter{\hbox{\includegraphics{main-figure\thecurrentfigure}}}
  \addtocounter{currentfigure}{1}
  \else
  \expandafter\startqcirc\BODY
  \checkforemp
  \pgfonlayer{background}
    \draw[opacity=0,white] ([xshift=-.15cm,yshift=.3cm] qc00) -- ([yshift=-.35cm,xshift=-.15cm] qc\theqcircrow0);
    \foreach \row in {0,...,\theqcircrow} {%
      \draw ([xshift=-.5cm] $(qc\row0)-(\qcircadjl,0)$) -- ([xshift=.5cm] $(qc\row\theqcirccol)+(\qcircadjr,0)$);
    }
  \endpgfonlayer
  \coordinate (c) at ($(qc00)!.5!(qc\theqcircrow0)$);
  \end{tikzpicture}
  \fi
}

\newif\ifinsertemp \insertemptrue
\def\hasgate{\insertempfalse}
\def\checkforemp{\ifinsertemp\emp\else\insertemptrue\fi}

\newcounter{qcirccol}
\newcounter{qcircrow}
\newdimen\qcircadjl
\newdimen\qcircadjr
\newdimen\qcircx \qcircx=1cm
\newdimen\qcircy \qcircy=.65cm
\def\qcircscale{1.0}

\def\pc{\hasgate\draw[fill=black] (\theqcirccol,\theqcircrow) coordinate (qc\theqcircrow\theqcirccol) circle (.1cm);}
\def\nc{\hasgate\draw[fill=white] (\theqcirccol,\theqcircrow) coordinate (qc\theqcircrow\theqcirccol) circle (.1cm);}

\def\casec{%
  \hasgate
  \coordinate (qc\theqcircrow\theqcirccol) at (\theqcirccol,\theqcircrow);
  \draw[fill=white] ([xshift=-.075cm] \theqcirccol,\theqcircrow) circle (.1cm);
  \draw[fill=black] ([xshift=.075cm] \theqcirccol,\theqcircrow) circle (.1cm);
}

\def\g#1{\hasgate\node[draw,fill=white,inner sep=1.5pt] (qc\theqcircrow\theqcirccol node) at (\theqcirccol,\theqcircrow) {$#1$}; \coordinate (qc\theqcircrow\theqcirccol) at (qc\theqcircrow\theqcirccol node.center);}

\def\caseg#1#2{\hasgate\node[draw,fill=white,inner sep=1.5pt] (qc\theqcircrow\theqcirccol node) at (\theqcirccol,\theqcircrow) {\raise3pt\hbox{$\scriptstyle #1$}\lower3pt\hbox{$\scriptstyle #2$}}; \draw (qc\theqcircrow\theqcirccol node.north east) -- (qc\theqcircrow\theqcirccol node.south west); \coordinate (qc\theqcircrow\theqcirccol) at (qc\theqcircrow\theqcirccol node.center);}

\def\targ{%
  \hasgate
  \draw[fill=white] (\theqcirccol,\theqcircrow) coordinate (qc\theqcircrow\theqcirccol) circle (.2cm);
  \draw ($(\theqcirccol,\theqcircrow)-(.2cm,0)$) -- ++(right:.4cm);
  \draw ($(\theqcirccol,\theqcircrow)-(0,.2cm)$) -- ++(up:.4cm);
}

\def\swap#1#2#3{%
  \draw[white] ($(qc#2#1)-(7pt,0)$) coordinate (a) -- ($(qc#2#1)+(7pt,0)$) coordinate (b);
  \draw[white] ($(qc#3#1)-(7pt,0)$) coordinate (c) -- ($(qc#3#1)+(7pt,0)$) coordinate (d);
  \draw (a) -- (d) (b) -- (c);
}

\def\con#1{%
  \pgfonlayer{background}
  \draw (qc\theqcircrow\theqcirccol) -- (qc#1\theqcirccol);
  \endpgfonlayer
}

\def\emp{\coordinate (qc\theqcircrow\theqcirccol) at (\theqcirccol,\theqcircrow);}

\pdfoutput=1

\begin{document}

\title{On quantum circuits employing roots of the Pauli matrices}
\author{Mathias Soeken}
\affiliation{Institute of Computer Science, University of Bremen, Germany}
\author{D. Michael Miller}
\affiliation{Department of Computer Science, University of Victoria, BC, Canada}
\author{Rolf Drechsler}
\affiliation{Institute of Computer Science, University of Bremen, Germany}

\begin{abstract}
  The Pauli matrices are a set of three $2\times 2$ complex Hermitian, unitary matrices.  In this article, we investigate the relationships between certain roots of the Pauli matrices and how gates implementing those roots
  are used in quantum circuits.  Techniques for simplifying such circuits are given.
  In particular, we show how those techniques can be used to find a circuit of
  Clifford+$T$ gates starting from a circuit composed of gates from the well studied NCV library.
\end{abstract}

\maketitle

\section{Introduction}
The well-studied NCV quantum gate library \cite{NC:2000} contains the gates: NOT~($X$),
controlled NOT~(CNOT), and both the single controlled square root of NOT as well as
its adjoint, denoted~$V$ and~$V^\dagger$, respectively.
In their seminal paper \cite{BBC+:1995},  Barenco \emph{et al}.~presented a
general result that as one instance shows how a classical reversible Toffoli
gate can be realized by five NCV gates as follows:
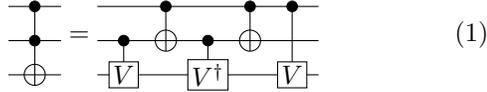
\begin{equation}
  \label{eq:barenco}
  \def\qcircscale{0.7}
  \qcircx=.8cm
  \begin{qcirc} \pc \cr \pc \cr \targ\con0 \end{qcirc}
  =
  \begin{qcirc}
    \emp  \> \pc        \> \emp \> \pc        \> \pc  \cr
    \pc   \> \targ\con0 \> \pc  \> \targ\con0 \> \emp \cr
    \g{V}\con1 \> \emp       \> \g{V^\dagger}\con1 \> \emp \> \g{V}\con0
  \end{qcirc}
\end{equation}
This result has been used in several works on synthesis and optimization
of quantum circuits~\cite{HSY+:2006,GWD+:2009,SM:2012}.

Another quantum gate library called Clifford+$T$~\cite{NC:2000}, which consists of
the controlled NOT, the phase gate $S$ and the Hadamard gate~$H$, plus the
$T$ gate
\begin{equation}
  \label{eq:t-gate}
  T = \left(\begin{array}{cc} 1 & 0 \\ 0 & e^{\frac{i\pi}{4}} \end{array}\right),
\end{equation}
and the adjoint gates $S^\dagger$ and $T^\dagger$ has also received considerable
attention.
The Clifford+$T$ library has the advantage over the NCV library
with respect to fault-tolerant computing~\cite{BCL+:2006}.

Recently different synthesis results based on the Clifford+$T$ gate library have
been presented~\cite{BS:2012,AMM+:2013,Sel:2013}.
One common aim in these works is to reduce the so-called
$T$-\emph{depth}, \emph{i.e.}~the number of $T$-\emph{stages} where each stage consists of one
or more~$T$ or~$T^\dagger$ gates that can operate simultaneously on separate
qubits.
In~\cite{AMM+:2013} the authors describe a search-based algorithm that finds
optimal circuit realizations with respect to their $T$-depth.
One of their circuits that realizes a Toffoli gate is~\cite[Fig.~13]{AMM+:2013}:
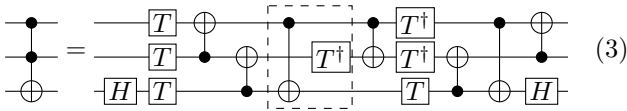
\begin{equation}
  \label{eq:amy-toffoli}
  \def\qcircscale{0.7}
  \qcircx=.8cm
  \begin{qcirc} \pc \cr \pc \cr \targ\con0 \end{qcirc}
  =
  % \begin{qcirc}
  % \g{T^\dagger} \> \targ \> \g{T} \> \targ \> \g{T^\dagger} \> \targ \> \g{T} \> \targ \cr
  % \g{T^\dagger} \> \emp \> \pc \> \pc\con0 \> \pc \> \emp \> \emp \> \pc\con0 \cr
  % \g{H} \> \pc\con0 \> \targ\con1 \> \g{T} \> \targ\con1 \> \pc\con0 \> \g{T^\dagger} \> \g{H}
  % \end{qcirc}.
  \begin{qcirc}
  \emp \> \g{T} \> \targ \> \emp \> \pc \> \emp \> \pc \> \g{T^\dagger} \> \emp \> \pc \> \targ \cr
  \emp \> \g{T} \> \pc\con0 \> \targ \> \emp \> \g{T^\dagger} \> \targ\con0 \> \g{T^\dagger} \> \targ \> \emp \> \pc\con0 \cr
  \g{H} \> \g{T} \> \emp \> \pc\con1 \> \targ\con0 \> \emp \> \emp \> \g{T} \> \pc\con1 \> \targ\con0 \> \g{H}
  \draw[dashed] ($(qc04)-(.5,.5)$) rectangle ++(2,3);
  \end{qcirc}
\end{equation}
This circuit has a $T$-depth of~$3$ and a total depth of~$10$.
(Note that the gates surrounded by the dashed rectangle together have a depth
of~$1$ and are drawn in sequence only for clarity.)
The approach in \cite{AMM+:2013} produces optimal circuits but since the
technique's complexity is exponential, it is only applicable to small circuits.

The NOT operation is described by the Pauli $X$ matrix~\cite{CM:1929}.
Note that the controlled NOT could be drawn as a controlled $X$ but we use
the normal convention of a $\oplus$ as shown in the figures above.
The matrices for the $V$ and $V^\dagger$ operations as used in \eqref{eq:barenco}
are square roots of the Pauli $X$ matrix.
Similarly, the~$T$ gate operation is given by a matrix (\ref{eq:t-gate}) that is the fourth root
of the Pauli~$Z$ matrix.

The use of gates associated with different Pauli matrices and their roots within the same circuit, as illustrated in~\eqref{eq:amy-toffoli}, motivated us to explore the relations between  the Pauli
matrices and their roots, and to investigate how the associated gates can be used in constructing quantum circuits.
This article presents our findings and demonstrates their applicability in
\emph{deriving} the optimal circuits from~\cite{AMM+:2013} from known NCV circuits rather than using
exhaustive search techniques.

\section{Preliminaries}
% In this section, we will derive some identities that are used in the remainder
% of the paper in order to prove the main result.
% The identities relate the roots of the Pauli matrices to the rotation matrices.
% In fact, the roots of the Pauli matrices will be expressed in terms of rotation
% matrices which are in turn expressed using the Pauli matrices.
The three \emph{Pauli matrices}~\cite{CM:1929} are given by
\begin{equation}
  \sigma_1 = \left(\begin{array}{cr} 0 & 1 \\ 1 & 0 \end{array}\right), \;
  \sigma_2 = \left(\begin{array}{cr} 0 & -i \\ i & 0 \end{array}\right), \;
  \sigma_3 = \left(\begin{array}{cr} 1 & 0 \\ 0 & -1 \end{array}\right)
\end{equation}
The alternate naming~$X=\sigma_1$, $Y=\sigma_2$, and $Z=\sigma_3$ is often used
and we use it whenever we refer to a specific Pauli matrix.

Matrices describing \emph{rotations} around the three axes of the Bloch sphere
are given by
\begin{equation}
  R_a(\theta)
  = \cos{\textstyle\frac{\theta}{2}}I-i\sin{\textstyle\frac{\theta}{2}}\sigma_a
\end{equation}
where~$a\in \{1,2,3\}$~with~$\theta$ being the rotation angle and~$I$ the
identity matrix~\cite{NC:2000}.
Each Pauli matrix specifies  a half turn $(180^{\circ})$ rotation around a particular axis up to a
global phase, \emph{i.e.}
\begin{equation}
  \label{eq:pauli-as-rotation}
  \sigma_a=e^{\frac{i\pi}{2}}R_a(\pi).
\end{equation}
The conjugate transpose of~$R_a(\theta)$ is found by negating the
angle~$\theta$ or by multiplying it by another Pauli matrix~$\sigma_b$ from both
sides, \emph{i.e.}
\begin{equation}
  \label{eq:rotation-dagger}
  R^\dagger_a(\theta)=R_a(-\theta)=\sigma_bR_a(\theta)\sigma_b
\end{equation}
where~$a\neq b$.  Note that it does not matter which of the two possible $\sigma_b$ is used.
Since~$\sigma_b$ is Hermitian, we also have
$R_a(\theta)=\sigma_bR_a^\dagger(\theta)\sigma_b$.

Rotation matrices are additive with respect to their angle,
\emph{i.e.}~$R_a(\theta_1)R_a(\theta_2)=R_a(\theta_1+\theta_2)$, so one can derive
the~$k^{\rm th}$ root of the Pauli matrices as well as their conjugate transpose
from~\eqref{eq:pauli-as-rotation} as
\begin{equation}
  \label{eq:roots-as-rotation}
  \sqrt[k]{\sigma_a}=e^{\frac{i\pi}{2k}}R_a\left(\textstyle\frac{\pi}{k}\right)
  \quad
  \text{and}
  \quad
  \sqrt[k]{\sigma_a}^\dagger
  =
  e^{-\frac{i\pi}{2k}}R^\dagger_a\left(\textstyle\frac{\pi}{k}\right).
\end{equation}
For brevity, we term these matrices the \emph{Pauli roots}.

Using~\eqref{eq:roots-as-rotation} the rotation matrices can also be
expressed in terms of the roots of the Pauli matrices, \emph{i.e.}
\begin{equation}
  \label{eq:rotation-as-roots}
  R_a\left(\textstyle\frac{\pi}{k}\right)=e^{-\frac{i\pi}{2k}}\sqrt[k]{\sigma_a}
  \quad\text{and}\quad
  R^\dagger_a\left(\textstyle\frac{\pi}{k}\right)=
  e^{\frac{i\pi}{2k}}\sqrt[k]{\sigma_a}^\dagger.
\end{equation}
Consequently, we can derive
\begin{equation}
  \label{eq:root-as-dagger}
  \begin{array}{rcl}
  \sqrt[k]{\sigma_a}
  & \stackrel{\mathrm{(\ref{eq:roots-as-rotation})}}{=} &
  e^{\frac{i\pi}{2k}}R_a\left(\textstyle\frac{\pi}{k}\right)
  \stackrel{\mathrm{(\ref{eq:rotation-dagger})}}{=}
  e^{\frac{i\pi}{2k}}\sigma_bR^\dagger_a(\theta)\sigma_b \\
  & \stackrel{\mathrm{(\ref{eq:rotation-as-roots})}}{=} &
  e^{\frac{i\pi}{k}}\sigma_b\sqrt[k]{\sigma_a}^\dagger\sigma_b
  \end{array}
\end{equation}
and analogously~$\sqrt[k]{\sigma_a}^\dagger=
e^{-\frac{i\pi}{k}}\sigma_b\sqrt[k]{\sigma_a}\sigma_b$ for~$a\neq b$.
This last equation will be a key element in showing how control lines can be
removed from controlled gates by using roots of higher degree.
A number in brackets above an equal sign indicates the identity applied.

\emph{Translation} from one Pauli matrix to another is given by
\begin{equation}
  \label{eq:pauli-to-pauli}
  \sigma_a=\rho_{ab}\sigma_b\rho_{ab}
\end{equation}
where
\begin{equation}
  \label{eq:translation-matrix}
  \begin{aligned}
    \rho_{ab}=\rho_{ba}&=\tfrac{1}{\sqrt2}(\sigma_a+\sigma_b) \\
    &=e^{\frac{i\pi}{2}}R_a\left(\tfrac{\pi}{2}\right)
    R_b\left(\tfrac{\pi}{2}\right)
    R_a\left(\tfrac{\pi}{2}\right)
  \end{aligned}
\end{equation}
with~$a\neq b$ are \emph{translation matrices}.
Further, we define~$\rho_{aa}=I$.
Eq.~\eqref{eq:pauli-to-pauli} can be extended to the Pauli roots giving
\begin{equation}
  \label{eq:pauli-root-to-pauli-root}
  \sqrt[k]{\sigma_a}=\rho_{ab}\sqrt[k]{\sigma_b}\rho_{ab}.
\end{equation}
% i.e.
% \begingroup
%   \label{eq:translation-matrices-precise}
%   \def\entry#1{\tfrac{#1}{\sqrt2}}
%   $\rho_{12}=\mbox{\scriptsize$\left(\begin{array}{cc} 0 & \entry{1-i} \\ \entry{1+i} & 0 \end{array}\right)$}$,
%   $\rho_{13}=H=\mbox{\scriptsize$\left(\begin{array}{cr} \entry{1} & \entry{1} \\ \entry{1} & -\entry{1} \end{array}\right)$}$, and
%   $\rho_{23}=\mbox{\scriptsize$\left(\begin{array}{cr} \entry{1} & -\entry{i} \\ \entry{i} & -\entry{1} \end{array}\right)$}$.
% \endgroup

The well-known \emph{Hadamard} operator is given by the normalized
$2\times2$ matrix
\begin{equation}
  \label{eq:hadamard}
  H=\frac{1}{\sqrt{2}}\left(\begin{array}{cr} 1 & 1 \\ 1 & -1 \end{array}\right)
\end{equation}
which can also be expressed as
\begin{equation}
  \label{eq:H-as-lin-combination}
  H=\frac{1}{\sqrt{2}}(X+Z)
\end{equation}
or in terms of rotation matrices as
\begin{equation}
  \label{eq:H-as-rotations}
  H=e^{\frac{i\pi}{2}}
  R_1\left(\tfrac{\pi}{2}\right)
  R_3\left(\tfrac{\pi}{2}\right)
  R_1\left(\tfrac{\pi}{2}\right).
\end{equation}
In fact, $H = \rho_{1\,2} = \rho_{2\,1}$ so the Hadamard operator can be used to
translate between~$X$ and~$Z$, \emph{i.e.}
\begin{equation}
  \label{eq:x-to-z}
  Z=HXH \quad\text{and}\quad X=HZH
\end{equation}

Given a unitary~$2^n\times 2^n$ matrix~$U$, called the \emph{target operation}, we
define four controlled operations
\begin{equation}
  \label{eq:controlled-gates-pictures}
  \begin{qcirc} \pc \cr \g{U}\con0 \end{qcirc},\;
  \begin{qcirc} \g{U} \cr \pc\con0 \end{qcirc},\;
  \begin{qcirc} \nc \cr \g{U}\con0 \end{qcirc},\;\text{ and }\;
  \begin{qcirc} \g{U} \cr \nc\con0 \end{qcirc}
\end{equation}
which are described by
\begin{equation}
  \label{eq:positive-control}
  \begin{array}{l}
    C_1(U)=\|0>\!\<0|\otimes I_{2^n}+\|1>\!\<1|\otimes U, \\[3pt]
    C_2(U)=I_{2^n}\otimes\|0>\!\<0|+U\otimes\|1>\!\<1|, \\[3pt]
    C_1^-(U)=\|0>\!\<0|\otimes U+\|1>\!\<1|\otimes I_{2^n}, \text{ and}\\[3pt]
    C_2^-(U)=U\otimes\|0>\!\<0|+I_{2^n}\otimes\|1>\!\<1|,
  \end{array}
\end{equation}
respectively,
where~$\otimes$ denotes the Kronecker product,~$I_{2^n}$ denotes the
$2^n\times2^n$ identity matrix, and~$\|\cdot>$ and~$\<\cdot|$ is Dirac's bra-ket
notation~\cite{Dirac:1939}.
The latter two operations are referred to as negative controlled operations.
%Negative control operations~$C_1^-(U)$ and~$C_2^-(U)$ are obtained by
%swapping~$I$ and~$U$ in Eq.~(\ref{eq:positive-control}).
As examples, the CNOT gate is~$C_1(X)$ and the Toffoli gate is~$C_1(C_1(X))$.
Notice that adjacent controlled operations with the same control preserve
matrix multiplication, \emph{i.e.}
\begin{equation}
  \label{eq:matrix-multiplication-controlled}
  \def\qcircscale{0.9}
  \qcircx=.8cm
  \begin{qcirc} \pc \> \pc \cr \g{U_2}\con0 \> \g{U_1}\con0 \end{qcirc}=
  \begin{qcirc} \pc \cr \g{U_1U_2}\con0 \end{qcirc}\quad\text{and}\quad
  \begin{qcirc} \nc \> \nc \cr \g{U_2}\con0 \> \g{U_1}\con0 \end{qcirc}=
  \begin{qcirc} \nc \cr \g{U_1U_2}\con0 \end{qcirc}
\end{equation}
and gates with opposite control polarities commute, \emph{i.e.}
\begin{equation}
  \label{eq:matrix-commute-controlled}
  \def\qcircscale{0.9}
  \qcircx=.8cm
  \begin{qcirc} \pc \> \nc \cr \g{U_2}\con0 \> \g{U_1}\con0 \end{qcirc}=
  \begin{qcirc} \nc \> \pc \cr \g{U_1}\con0 \> \g{U_2}\con0 \end{qcirc}
\end{equation}
Finally, if the same target operation is controlled with both polarities, the target
operation does not need be controlled, \emph{i.e.}
\begin{equation}
  \label{eq:matrix-both}
  \def\qcircscale{0.9}
  \qcircx=.8cm
  \begin{qcirc} \pc \> \nc \cr \g{U}\con0 \> \g{U}\con0 \end{qcirc}=
  \begin{qcirc} \nc \> \pc \cr \g{U}\con0 \> \g{U}\con0 \end{qcirc}=
  \begin{qcirc} \emp \cr \g{U} \end{qcirc}
\end{equation}
The identities in~\eqref{eq:matrix-multiplication-controlled}--\eqref{eq:matrix-both}
hold for~$C_2$ and~$C_2^-$ analogously.
Eq.~\eqref{eq:matrix-commute-controlled} describes a circuit which performs
one of two different target operations depending on the value assigned to the
control line.
Let us consider the special case
\begin{equation}
  \label{eq:case-matrix}
  C_1^-(U_1)C_1(U_2)=\|0>\!\<0|\otimes U_1+\|1>\!\<1|\otimes U_2
\end{equation}
which performs~$U_1$ or~$U_2$ if the control line is negative or positive,
respectively.
We call this matrix a~\emph{case gate} and employ a special gate representation shown on the left-hand side of:
\begin{equation}
  \label{eq:case-gate}
  \begin{qcirc} \casec \cr \caseg{U_1}{U_2}\con0 \end{qcirc}
  =
  \begin{qcirc} \nc \> \pc \cr \g{U_1}\con0 \> \g{U_2}\con0 \end{qcirc}
\end{equation}

Another interesting circuit equality applies to the Pauli~$Z$
gate and its roots.
A positively controlled gate can be flipped without changing its functionality,
\emph{i.e.}
\begin{equation}
  \label{eq:switch-z-roots}
  C_1(\sqrt[k]{Z})=
  \begin{qcirc} \pc \cr \g{\sqrt[k]{Z}}\con0 \end{qcirc}=
  \begin{qcirc} \g{\sqrt[k]{Z}}\cr \pc\con0 \end{qcirc}=
  C_2(\sqrt[k]{Z}).
\end{equation}
Using the translation gates, a circuit identity for the other
Pauli roots can be derived:
\begin{displaymath}
  \begin{qcirc} \pc \cr \g{\sqrt[k]{\sigma_a}}\con0 \end{qcirc}
  \stackrel{(\ref{eq:pauli-root-to-pauli-root})}{=}
  \begin{qcirc} \emp \> \pc \> \emp \cr \g{\rho_{a3}} \> \g{\sqrt[k]{Z}}\con0 \> \g{\rho_{a3}} \end{qcirc}
\end{displaymath}
\begin{equation}
  \label{eq:turn-pauli-with-hadamard}
  \stackrel{(\ref{eq:switch-z-roots})}{=}
  \begin{qcirc} \emp \> \g{\sqrt[k]{Z}} \> \emp \cr \g{\rho_{a3}} \> \pc\con0 \> \g{\rho_{a3}} \end{qcirc}
  \stackrel{(\ref{eq:pauli-root-to-pauli-root})}{=}
  \begin{qcirc} \g{\rho_{a3}} \> \g{\sqrt[k]{\sigma_a}} \> \g{\rho_{a3}} \cr \g{\rho_{a3}} \> \pc\con0 \> \g{\rho_{a3}} \end{qcirc}
\end{equation}
where~$\rho_{3\,3}=I$.

Since~$\sqrt[k]{Z}\otimes I=C_1(e^{\frac{i\pi}{k}}I)$
and
\[C_1(U)C_1(e^{\frac{i\pi}{k}}I)=C_1(e^{\frac{i\pi}{k}}U)=C_1(e^{\frac{i\pi}{k}}I)C_1(U)\]
a root of the Pauli $Z$ gate can be moved across a positive control, \emph{i.e.}
\begin{equation}
  \label{eq:move-t-over-control}
  \begin{qcirc} \g{\sqrt[k]{Z}} \> \pc \cr \> \g{U}\con0 \end{qcirc}
  =
  \begin{qcirc} \pc \> \g{\sqrt[k]{Z}} \cr \g{U}\con0 \> \end{qcirc}
\end{equation}

Lastly, in the remainder of the paper we will often make use of a circuit
identity given in~\cite[Rule D7]{ST:2013}:
\begin{equation}
  \label{eq:cnot-rule}
  \def\qcircscale{0.8}
  \qcircx=.8cm
  \begin{qcirc} \pc \> \cr \targ\con0 \> \pc \cr \> \targ\con1 \end{qcirc}
  =
  \begin{qcirc} \> \pc \> \pc \cr \pc \> \targ\con0 \> \cr \targ\con1 \> \> \targ\con0 \end{qcirc}
\end{equation}

\section{Mapping schemes for synthesis}
Many of the mappings for translating a given quantum circuit into a circuit
in terms of a restricted gate library in fact reduce the
number of controls in a controlled gate, as shown for example in \eqref{eq:barenco}.
We follow this approach and will provide generic rules in terms of general Pauli roots.

First, we show how a single controlled Pauli root~$\sqrt[k]{\sigma_a}$ gate can be
mapped to a circuit involving uncontrolled Pauli root gates by doubling the index,
\emph{i.e.}~$\sqrt[2k]{\sigma_a}$.
We start by proving the following lemma.
\begin{lemma}
  \label{lemma:case-circuit}
  The following circuit equality holds for~$a\neq b$:
\begin{equation}
  \label{eq:lemma-case-circuit}
  \begin{qcirc} \casec \cr
         \caseg{\sqrt[k]{\sigma_a}^\dagger}{\sqrt[k]{\sigma_a}} \con0
  \qcircadjl=.5cm
  \qcircadjr=.5cm
  \end{qcirc}
  =
  \begin{qcirc}
    \pc \> \g{\root k\of Z} \> \pc \cr
    \g{\sigma_b}\con0 \> \g{\sqrt[k]{\sigma_a}^\dagger} \> \g{\sigma_b}\con0
  \end{qcirc}
\end{equation}
\end{lemma}
\begin{proof}
  Based on the special structure
of~$\sqrt[k]{Z}=\begin{pmatrix}1 & 0 \\ 0 & e^{\frac{i\pi}{k}}\end{pmatrix}$, one
obtains
\[
  \def\qcircscale{.8}
  \qcircy=.8cm
  \begin{qcirc}
    \pc \> \g{\sqrt[k]{Z}} \> \pc \cr
    \g{\sigma_b}\con0 \> \g{\sqrt[k]{\sigma_a}^\dagger} \> \g{\sigma_b}\con0
  \end{qcirc}
  \stackrel{(\ref{eq:roots-as-rotation})}{=}
  \begin{qcirc}
    \pc \> \nc \> \emp \> \pc \> \emp \> \pc \cr
    \g{\sigma_b}\con0 \> \g{\sqrt[k]{\sigma_a}^\dagger}\con0 \>
    \emp \> \g{e^{\frac{i\pi}{k}}\sqrt[k]{\sigma_a}^\dagger}\con0 \>
    \emp \> \g{\sigma_b}\con0
  \end{qcirc}
\]
\[
  \def\qcircscale{.8}
  \qcircy=.8cm
  \stackrel{~(\ref{eq:matrix-multiplication-controlled},\ref{eq:matrix-commute-controlled})}{=}
  \begin{qcirc}
    \nc \> \emp \> \pc \cr
    \g{\sqrt[k]{\sigma_a}^\dagger}\con0 \>
    \emp \> \g{\sigma_be^{\frac{i\pi}{k}}\sqrt[k]{\sigma_a}^\dagger\sigma_b}\con0
  \qcircadjl=.25cm
  \qcircadjr=.85cm
  \end{qcirc}
  \stackrel{~(\ref{eq:root-as-dagger})}{=}
  \begin{qcirc}
    \nc \> \emp \> \pc \cr
    \g{\sqrt[k]{\sigma_a}^\dagger}\con0 \> \emp \>
    \g{\sqrt[k]{\sigma_a}}\con0
  \qcircadjl=.25cm
  \qcircadjr=.25cm
  \end{qcirc}
\]
which proves~\eqref{eq:lemma-case-circuit}.
\end{proof}

\begin{corollary}
  Analogously to Lemma~\ref{lemma:case-circuit} we can derive
\begin{equation}
  \label{eq:lemma-case-circuit2}
  \begin{qcirc} \casec \cr
         \caseg{\sqrt[k]{\sigma_a}}{\sqrt[k]{\sigma_a}^\dagger} \con0
  \qcircadjl=.5cm
  \qcircadjr=.5cm
  \end{qcirc}
  =
  \begin{qcirc}
    \pc \> \g{\sqrt[k]{Z}^\dagger} \> \pc \cr
    \g{\sigma_b}\con0 \> \g{\sqrt[k]{\sigma_a}} \> \g{\sigma_b}\con0
  \end{qcirc}
\end{equation}
\end{corollary}
Lemma~\ref{lemma:case-circuit} can be used to prove the following theorem that shows how
to remove a control line from a controlled Pauli root operation.
\begin{theorem}
  \label{theo:remove-one-line}
If~$a\neq b$, then the following circuit equality holds:
\begin{equation}
  \label{eq:theo-remove-one-line}
  \begin{qcirc}
  \pc \cr \g{\sqrt[k]{\sigma_a}} \con0
  \qcircadjl=.25cm
  \qcircadjr=.25cm
  \end{qcirc}
  =
  \begin{qcirc}
    \> \pc \> \g{\sqrt[2k]{Z}} \> \pc \cr
    \g{\sqrt[2k]{\sigma_a}} \> \g{\sigma_b}\con0 \>
    \g{\sqrt[2k]{\sigma_a}^\dagger} \> \g{\sigma_b}\con0
    \qcircadjl=.25cm
  \end{qcirc}
\end{equation}
\end{theorem}
\begin{proof}
We have
\[
  \def\qcircscale{.8}
  \qcircy=.8cm
  \qcircx=1.1cm
  \begin{qcirc}
    \> \pc \> \g{\sqrt[2k]{Z}} \> \pc \cr
    \g{\sqrt[2k]{\sigma_a}} \> \g{\sigma_b}\con0 \>
    \g{\sqrt[2k]{\sigma_a}^\dagger} \> \g{\sigma_b}\con0
    \qcircadjl=.25cm
  \end{qcirc}
  \stackrel{\eqref{eq:lemma-case-circuit}}{=}
  \qcircx=.7cm
  \begin{qcirc}
      \emp \> \emp \> \nc \> \emp \> \pc \cr
      \g{\sqrt[2k]{\sigma_a}} \> \emp \>
      \g{\sqrt[2k]{\sigma_a}^\dagger}\con0 \> \emp \>
      \g{\sqrt[2k]{\sigma_a}}\con0
      \qcircadjl=.25cm
      \qcircadjr=.25cm
  \end{qcirc}
\]
\[
  \def\qcircscale{.8}
  \qcircy=.8cm
  \qcircx=.7cm
  \stackrel{(\ref{eq:matrix-both})}{=}
  \begin{qcirc}
      \nc \> \emp \> \pc \> \emp \> \nc \> \emp \> \pc \cr
      \g{\sqrt[2k]{\sigma_a}}\con0 \> \emp \> \g{\sqrt[2k]{\sigma_a}}\con0
      \> \emp \>
      \g{\sqrt[2k]{\sigma_a}^\dagger}\con0 \> \emp \>
      \g{\sqrt[2k]{\sigma_a}}\con0
      \qcircadjl=.25cm
      \qcircadjr=.25cm
  \end{qcirc}
\]
\[
  \def\qcircscale{.8}
  \qcircy=.8cm
  \qcircx=1.2cm
  \stackrel{~(\ref{eq:matrix-multiplication-controlled},\ref{eq:matrix-commute-controlled})}{=}
  \begin{qcirc}
      \nc \> \emp \> \pc \cr
      \g{\sqrt[2k]{\sigma_a}\sqrt[2k]{\sigma_a}^\dagger}\con0
      \> \emp \>
      \g{\sqrt[2k]{\sigma_a}\sqrt[2k]{\sigma_a}}\con0
      \qcircadjl=.95cm
      \qcircadjr=.95cm
  \end{qcirc}
\]
\[
  \def\qcircscale{.8}
  \qcircy=.8cm
  =
  \begin{qcirc}
    \nc \> \pc \cr
    \g{I}\con0 \> \g{\root k\of{\sigma_a}}\con0
    \qcircadjr=.25cm
  \end{qcirc}
  =
  \begin{qcirc}
    \pc \cr
    \g{\root k\of{\sigma_a}}\con0
    \qcircadjl=.25cm
    \qcircadjr=.25cm
  \end{qcirc}
\]
which concludes the proof.
\end{proof}
It follows from the proof that the first gate on the right-hand side of~(\ref{eq:theo-remove-one-line})
commutes with the other three gates treated as a single group.
Using Theorem~\ref{theo:remove-one-line} we can for example derive the identity shown
in~\cite[Fig.~5(b)]{AMM+:2013} where~$S=\sqrt{Z}$:
\[
  \def\qcircscale{.7}
  \qcircx=.68cm
  \begin{qcirc} \pc \cr \g{S}\con0 \end{qcirc}
  \!\stackrel{(\ref{eq:switch-z-roots})}{=}\!
  \begin{qcirc} \g{S} \cr \pc\con0 \end{qcirc}
  \!\stackrel{(\ref{eq:theo-remove-one-line})}{=}\!
  \begin{qcirc} \targ \> \g{T^\dagger} \> \targ \> \g{T} \cr
         \pc\con0 \> \g{T} \> \pc\con0 \> \emp
  \end{qcirc}
  \!\stackrel{\eqref{eq:move-t-over-control}}{=}\!
  \begin{qcirc} \targ \> \g{T^\dagger} \> \targ \> \g{T} \cr
         \pc\con0 \> \emp \> \pc\con0 \> \g{T}
  \end{qcirc}
\]
Note that a root of the Pauli~$Z$ gate can be moved across the control of
a CNOT but not over the target.

In~\cite[Lemma 6.1]{BBC+:1995} the circuit identity of which
\eqref{eq:barenco} is a particular case is defined for an arbitrary unitary
matrix~$U$ on the left-hand side and its square root on the right-hand side.
Given that Lemma and the translation gate equality in
\eqref{eq:pauli-root-to-pauli-root} we can define a new generic mapping that involves the Pauli roots.
\begin{theorem}
  \label{theo:barenco-extended}
  The following circuit equality holds:
\begin{equation}
  \label{eq:theo-barenco-extended}
  \def\qcircscale{.79}
  \begin{qcirc}
    \pc \cr \pc \cr \g{\sqrt[k]{\sigma_a}}\con0
    \qcircadjl=.25cm
    \qcircadjr=.25cm
  \end{qcirc}
  =
  \begin{qcirc}
    \emp \> \emp  \> \pc        \> \emp \> \pc        \> \pc  \> \emp \cr
    \emp \> \pc   \> \targ\con0 \> \pc  \> \targ\con0 \> \emp \> \emp \cr
    \g{\rho_{ab}} \>
    \g{\sqrt[2k]{\sigma_b}}\con1 \> \emp \>
    \g{\sqrt[2k]{\sigma_b}^\dagger}\con1 \> \emp \>
    \g{\sqrt[2k]{\sigma_b}}\con0 \>
    \g{\rho_{ab}}
    \qcircadjl=.25cm
    \qcircadjr=.25cm
  \end{qcirc}
\end{equation}
where~$\rho_{ab}=I$ if~$a=b$.
\end{theorem}
\begin{proof}
  Considering the circuit one gets from~\cite[Lemma 6.1]{BBC+:1995},
  the Pauli roots are adjusted by inserting translation gates around each of them.
  Since the translation gates are Hermitian they are involutory and therefore
  two translation gates in sequence cancel.
\end{proof}
One obtains~\eqref{eq:barenco} by setting~$a=b=k=1$ in~\eqref{eq:theo-barenco-extended}.
Note that the circuit identities in both \eqref{eq:theo-remove-one-line} and
\eqref{eq:theo-barenco-extended} can be extended by adding a control to the
left-hand side and controls to particular gates on the right-hand side.
Repeating this process enables the mapping of multiple controlled
gates.
As the above illustrates, one needs to double the degree of the root in order to remove
one control line when not using ancilla lines.
As a result, one needs operations~$\sqrt[2^n]{\sigma_a}$ in order to represent an
$n$-controlled Toffoli gate if the only allowed controlled gate is a CNOT.

By combining \eqref{eq:theo-barenco-extended} with \eqref{eq:theo-remove-one-line}
a powerful mapping scheme can be obtained for
Toffoli gates that is very flexible due to the number of degrees of freedom:
\begin{enumerate}
\itemsep0pt
\item In \eqref{eq:theo-remove-one-line} we can move the first gate to the end of the circuit.
\item In \eqref{eq:theo-remove-one-line} the Pauli $Z$ root can be moved
  over the controls on the upper line.
\item In \eqref{eq:theo-remove-one-line} the $C_1(\sigma_b)$ operations can
  be replaced by~$C_2(\sigma_b)$ operations, if~$b=3$.
\item In \eqref{eq:theo-barenco-extended} the first two lines can be
  swapped.
\item In \eqref{eq:theo-barenco-extended} the functionality is preserved if
  all controlled Pauli roots are replaced by their adjoint..
\item The last controlled Pauli root in \eqref{eq:theo-barenco-extended}
  can be freely moved within the $\rho_{ab}$ gates.
\end{enumerate}

By exploiting these options, it is possible to derive a $T$ gate realization for
a Toffoli gate of the same quality as the one shown in~\eqref{eq:amy-toffoli}
starting from the mapping in~\eqref{eq:barenco} instead of using expensive
search techniques.
We will demonstrate this by deriving a realization for a more general gate
%Next, we will derive a realization for a more general gate
\begin{equation}
  \label{eq:single-target-gate}
  \qcircx=.75cm
  \def\qcircscale{.7}
  \begin{qcirc}
    \emp \cr \emp \cr \targ\con0
    \draw[fill=white,rounded corners=2pt] ($(qc00)-(4.5pt,-3.5pt)$) coordinate (a) rectangle ($(qc10)+(4.5pt,-3.5pt)$) coordinate (b);
    \node[font=\footnotesize] at ($(a)!.5!(b)$) {$f$};
  \begin{scope}[every node/.style={font=\footnotesize,inner sep=1.5pt}]
    \node[left,xshift=-.3cm] at (qc00) {$x_1$};
    \node[left,xshift=-.3cm] at (qc10) {$x_2$};
    \node[left,xshift=-.3cm] at (qc20) {$x_3$};
    \node[right,xshift=.3cm] at (qc00) {$x_1$};
    \node[right,xshift=.3cm] at (qc10) {$x_2$};
    \node[right,xshift=.3cm] at (qc20) {$x_3\oplus f(x_1,x_2)$};
  \end{scope}
  \end{qcirc}
\end{equation}
in which the target line~$x_3$ is inverted with respect to a control
function~$f$ of a particular type.
The Toffoli gate in~\eqref{eq:amy-toffoli} is the special case where~$f=x_1x_2$.
Given a unitary matrix~$U$, in this construction, we will use the notation~$U^p$
for~$p\in\{0,1\}$ where~$U^0=U$ and~$U^1=U^\dagger$.

\begin{figure*}[t]
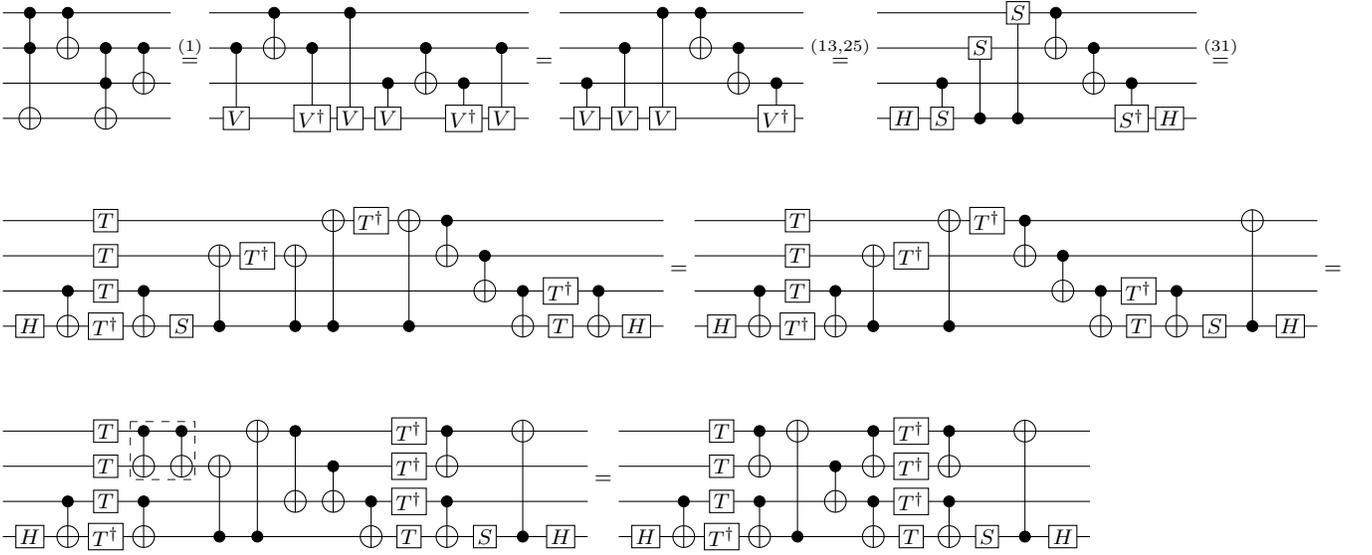

  \footnotesize
  \qcircx=.7cm
  \def\qcircscale{.72}
%  \centering
$
  \begin{qcirc}
    \pc \> \pc \> \emp \> \emp \cr
    \pc \> \targ\con0 \> \pc \> \pc \cr
    \emp \> \emp \> \pc \> \targ\con1 \cr
    \targ\con0 \> \emp \> \targ\con1 \> \emp
  \end{qcirc}
  \stackrel{(\ref{eq:barenco})}{=}
  \begin{qcirc}
    \emp \> \pc \> \emp \> \pc \> \emp \> \emp \> \emp \> \emp \cr
    \pc \> \targ\con0 \> \pc \> \emp \> \emp \> \pc \> \emp \> \pc \cr
    \emp \> \emp \> \emp \> \emp \> \pc \> \targ\con1 \> \pc \> \emp \cr
    \g{V}\con1 \> \emp \> \g{V^\dagger}\con1 \> \g{V}\con0 \> \g{V}\con2 \> \emp \> \g{V^\dagger}\con2 \> \g{V}\con1
  \end{qcirc}
  =
  \begin{qcirc}
    \emp \> \emp \> \pc \> \pc \> \emp \> \emp \cr
    \emp \> \pc \> \emp \> \targ\con0 \> \pc \> \emp \cr
    \pc \> \emp \> \emp \> \emp \> \targ\con1 \> \pc \cr
    \g{V}\con2 \> \g{V}\con1 \> \g{V}\con0 \> \emp \> \emp \> \g{V^\dagger}\con2
  \end{qcirc}
  \stackrel{(\ref{eq:pauli-root-to-pauli-root},\ref{eq:switch-z-roots})}{=}
  \begin{qcirc}
    \emp \> \emp \> \emp \> \g{S} \> \pc \> \emp \> \emp \> \emp \cr
    \emp \> \emp \> \g{S} \> \emp \> \targ\con0 \> \pc \> \emp \> \emp \cr
    \emp \> \pc \> \emp \> \emp \> \emp \> \targ\con1 \> \pc \> \emp \cr
    \g{H} \> \g{S}\con2 \> \pc\con1 \> \pc\con0 \> \emp \> \emp \> \g{S^\dagger}\con2 \> \g{H}
  \end{qcirc}
  \stackrel{(\ref{eq:theo-remove-one-line})}{=}
$\hfill\phantom{.}
\bigskip
\bigskip

$
  \begin{qcirc}
    \emp \> \emp \> \g{T} \> \emp \> \emp \> \emp \> \emp \> \emp \> \targ \> \g{T^\dagger} \> \targ \> \pc \> \emp \> \emp \> \emp \> \emp \> \emp \cr
    \emp \> \emp \> \g{T} \> \emp \> \emp \> \targ \> \g{T^\dagger} \> \targ \> \emp \> \emp \> \emp \> \targ\con0 \> \pc \> \emp \> \emp \> \emp \> \emp \cr
    \emp \> \pc \> \g{T} \> \pc \> \emp \> \emp \> \emp \> \emp \> \emp \> \emp \> \emp \> \emp \> \targ\con1 \> \pc \> \g{T^\dagger} \> \pc \> \emp \cr
    \g{H} \> \targ\con2 \> \g{T^\dagger} \> \targ\con2 \> \g{S} \> \pc\con1 \> \emp \> \pc\con1 \> \pc\con0 \> \emp \> \pc\con0 \> \emp \> \emp \> \targ\con2 \> \g{T} \> \targ\con2 \> \g{H}
  \end{qcirc}
  =
  \begin{qcirc}
  \emp \> \emp \> \g{T} \> \emp \> \emp \> \emp \> \targ \> \g{T^\dagger} \> \pc \> \emp \> \emp \> \emp \> \emp \> \emp \> \targ \> \emp \cr
    \emp \> \emp \> \g{T} \> \emp \> \targ \> \g{T^\dagger} \> \emp \> \emp \> \targ\con0 \> \pc \> \emp \> \emp \> \emp \> \emp \> \emp \> \emp \cr
    \emp \> \pc \> \g{T} \> \pc \> \emp \> \emp \> \emp \> \emp \> \emp \> \targ\con1 \> \pc \> \g{T^\dagger} \> \pc \> \emp \> \emp \> \emp \cr
    \g{H} \> \targ\con2 \> \g{T^\dagger} \> \targ\con2 \> \pc\con1 \> \emp \> \pc\con0 \> \emp \> \emp \> \emp \> \targ\con2 \> \g{T} \> \targ\con2 \> \g{S} \> \pc\con0 \> \g{H}
  \end{qcirc}
  =
$\hfill\phantom{.}
\bigskip
\bigskip

$
  \begin{qcirc}
  \emp \> \emp \> \g{T} \> \pc \> \pc \> \emp \> \targ \> \pc \> \emp \> \emp \> \g{T^\dagger} \> \pc \> \emp \> \targ \> \emp \cr
  \emp \> \emp \> \g{T} \> \targ\con0 \> \targ\con0 \> \targ \> \emp \> \emp \> \pc \> \emp \> \g{T^\dagger} \> \targ\con0 \> \emp \> \emp \> \emp \cr
  \emp \> \pc \> \g{T} \> \pc \> \emp \> \emp \> \emp \> \targ\con0 \> \targ\con1 \> \pc \> \g{T^\dagger} \> \pc \> \emp \> \emp \> \emp \cr
  \g{H} \> \targ\con2 \> \g{T^\dagger} \> \targ\con2 \> \emp \> \pc\con1 \> \pc\con0 \> \emp \> \emp \> \targ\con2 \> \g{T} \> \targ\con2 \> \g{S} \> \pc\con0 \> \g{H}
  \draw[dashed,gray!30!black] ($(qc03)-(7pt,-5pt)$) rectangle ($(qc14)+(7pt,-7pt)$);
  \end{qcirc}
  =
  \begin{qcirc}
  \emp \> \emp \> \g{T} \> \pc \> \targ \> \emp \> \pc \> \g{T^\dagger} \> \pc \> \emp \> \targ \> \emp \cr
  \emp \> \emp \> \g{T} \> \targ\con0 \> \emp \> \pc \> \targ\con0 \> \g{T^\dagger} \> \targ\con0 \> \emp \> \emp \> \emp \cr
  \emp \> \pc \> \g{T} \> \pc \> \emp \> \targ\con1 \> \pc \> \g{T^\dagger} \> \pc \> \emp \> \emp \> \emp \cr
  \g{H} \> \targ\con2 \> \g{T^\dagger} \> \targ\con2 \> \pc\con0 \> \emp \> \targ\con2 \> \g{T} \> \targ\con2 \> \g{S} \> \pc\con0 \> \g{H}
  \end{qcirc}
$\hfill\phantom{.}
  \caption{Derivation of the full adder}
  \label{fig:full-adder}
\end{figure*}

To begin, applying a modified version of the mapping in~\eqref{eq:barenco}
yields
\[
  \qcircx=.75cm
  \def\qcircscale{.7}
  \begin{qcirc}
    \emp  \> \pc    \> \emp  \> \pc    \> \pc \cr
    \pc \> \targ\con0 \> \pc \> \targ\con0 \> \emp  \cr
    \g{V^a}\con1 \> \emp \> \g{V^b}\con1 \> \emp \> \g{V^c}\con0
  \end{qcirc}
  \stackrel{(\ref{eq:theo-barenco-extended},\ref{eq:switch-z-roots})}{=}
  % \begin{qcirc}
  %   \emp \> \pc  \> \targ    \> \pc  \> \targ    \> \emp \> \emp \cr
  %   \emp \> \emp \> \pc\con0 \> \emp \> \pc\con0 \> \pc  \> \emp \cr
  %   \g{H} \> \g{S^a}\con0 \> \emp \> \g{S^b}\con0 \> \emp \> \g{S^c}\con1 \> \g{H}
  % \end{qcirc}
  \begin{qcirc}
    \emp \> \emp \> \pc    \> \emp \> \pc    \> \g{S^c} \> \emp \cr
    \emp \> \pc  \> \targ\con0 \> \pc  \> \targ\con0 \> \emp  \> \emp \cr
    \g{H} \> \g{S^a}\con1 \> \emp \> \g{S^b}\con1 \> \emp \> \pc\con0 \> \g{H}
  \end{qcirc}
\]
% According to Eq.~(\ref{eq:switch-z-roots}) we can switch the orientation of the
% $S$ gates leading to
% \[
%   \qcircx=.8cm
%   \def\qcircscale{.7}
%   \begin{qcirc}
%     \emp \> \g{S^a} \> \targ    \> \g{S^b} \> \targ    \> \emp \> \emp \cr
%     \emp \> \emp  \> \pc\con0 \> \emp  \> \pc\con0 \> \pc  \> \emp \cr
%     \g{H} \> \pc\con0 \> \emp \> \pc\con0 \> \emp \> \g{S^c}\con1 \> \g{H}
%   \end{qcirc}
% \]
Next, the controlled~$S$ gates are replaced according to~\eqref{eq:theo-remove-one-line}
into circuits with uncontrolled~$T$ gates.
We implicitly make use of the commutation property (first item in the list of
degrees of freedom) to enable the combination of gates.
\[
  \qcircx=.8cm
  \def\qcircscale{.7}
  % \begin{qcirc}
  %   \g{T^a} \> \targ \> \g{T^{\bar a}} \> \targ \> \targ \> \targ \> \g{T^{\bar b}} \> \targ \> \g{T^b} \> \targ \> \emp \> \emp \> \emp \> \emp \> \emp \cr
  %   \emp \> \emp \> \emp \> \emp \> \pc\con0 \> \emp \> \emp \> \emp \> \emp \> \pc\con0 \> \pc \> \g{T^c} \> \pc \> \emp \> \emp \cr
  %   \g{H} \> \pc\con0 \> \g{T^a} \> \pc\con0 \> \emp \> \pc\con0 \> \g{T^b} \> \pc\con0 \> \emp \> \emp \> \targ\con1 \> \g{T^{\bar c}} \> \targ\con1 \> \g{T^c} \> \g{H}
  % \draw[<->,thick,blue!80!white] ([yshift=-5pt] qc23) to[out=270,in=270] node[midway,below,inner sep=1pt] {\footnotesize cancel} ([yshift=-5pt] qc25);
  % \draw[<->,thick,blue!80!white] ([yshift=-8pt] qc22) to[out=270,in=270] node[midway,below,inner sep=1pt] {\footnotesize combine} ([yshift=-8pt] qc26);
  % \draw[->,thick,blue!80!white] ([xshift=-10pt] qc111) to[out=175,in=5] (qc10);
  % \end{qcirc}
  \begin{qcirc}
  \emp \> \emp \> \emp \> \emp \> \pc \> \emp \> \emp \> \emp \> \emp \> \pc \> \g{T^c} \> \targ \> \g{T^{\bar c}} \> \targ \> \emp \cr
  \emp \> \pc \> \g{T^a} \> \pc \> \targ\con0 \> \emp \> \pc \> \g{T^b} \> \pc \> \targ\con0 \> \emp \> \emp \> \emp \> \emp \> \emp \cr
  \g{H} \> \targ\con1 \> \g{T^{\bar a}} \> \targ\con1 \> \g{T^a} \> \g{T^b} \> \targ\con1 \> \g{T^{\bar b}} \> \targ\con1 \> \emp \> \emp \> \pc\con0 \> \g{T^c} \> \pc\con0 \> \g{H}
  \begin{scope}[<->,thick,gray!30!black,every node/.style={midway,below,inner sep=1pt,font=\footnotesize}]
    \draw[->] ([yshift=7pt] qc010) to[bend right=10] ([yshift=7pt] qc02);
    \draw ([yshift=-9pt] qc24) to[bend right=10] node {combine} ([yshift=-9pt] qc25);
    \draw (qc04) to[out=180,in=90] (qc13);
    \draw[->] ($(qc09)!.5!(qc19)$) to[bend right=15] (qc014);
  \end{scope}
  \end{qcirc}
\]
A~$T$ gate can move over a control line without changing the function, but not
over a target.
In order to move CNOT gates over other controlled gates we are making use
of~\eqref{eq:cnot-rule}.
The combined~$T^aT^b$ gate commutes with the block of three gates to the right to it.
\[
  \qcircx=.95cm
  \def\qcircscale{.7}
  % \begin{qcirc}
  % \g{T^a} \> \targ \> \g{T^{\bar a}} \> \targ \> \g{T^{\bar b}} \> \targ \> \g{T^b} \> \targ \> \emp \> \emp \> \emp \> \emp \> \emp \cr
  % \g{T^c} \> \emp \> \emp \> \pc\con0 \> \emp \> \emp \> \emp \> \pc\con0 \> \pc \> \emp \> \pc \> \emp \> emp \cr
  % \g{H} \> \pc\con0 \> \emp \> \emp \> \emp \> \pc\con0 \> \g{T^aT^b} \> \emp \> \targ\con1 \> \g{T^{\bar c}} \> \targ\con1 \> \g{T^c} \> \g{H}
  % \draw[thick,blue!80!white,decoration={brace},decorate] ([xshift=5pt,yshift=-7pt] qc210) coordinate (a) -- ([xshift=-5pt,yshift=-7pt] qc28) coordinate (b);
  % \draw[->,thick,blue!80!white] ([yshift=-2pt] $(a)!.5!(b)$) coordinate (c) to[out=270,in=270,looseness=.6] (c -| qc23);
  % \end{qcirc}
  \begin{qcirc}
  \emp \> \emp \> \g{T^c} \> \pc \> \pc \> \emp \> \emp \> \emp \> \targ \> \g{T^{\bar c}} \> \targ \> \pc \cr
  \emp \> \pc \> \g{T^a} \> \emp \> \targ\con0 \> \g{T^b} \> \pc \> \emp \> \emp \> \emp \> \emp \> \targ\con0 \cr
  \g{H} \> \targ\con1 \> \g{T^{\bar a}} \> \targ\con0 \> \emp \> \g{T^{\bar b}} \> \targ\con1 \> \g{T^aT^b} \> \pc\con0 \> \g{T^c} \> \pc\con0 \> \g{H}
  \begin{scope}[<->,thick,gray!30!black,every node/.style={midway,below,inner sep=1pt,font=\footnotesize}]
    \draw[->] ([yshift=7pt] qc15) to[bend left=10] ([yshift=7pt] qc19);
  \end{scope}
  \end{qcirc}
\]
This last step reduces the total depth by one.
\[
  \qcircx=.95cm
  \def\qcircscale{.7}
  \begin{qcirc}
  \emp \> \emp \> \g{T^c} \> \pc \> \pc \> \emp \> \emp \> \targ \> \g{T^{\bar c}} \> \targ \> \pc \cr
  \emp \> \pc \> \g{T^a} \> \emp \> \targ\con0 \> \pc \> \emp \> \emp \> \g{T^b} \> \emp \> \targ\con0 \cr
  \g{H} \> \targ\con1 \> \g{T^{\bar a}} \> \targ\con0 \> \g{T^{\bar b}} \> \targ\con1 \> \g{T^aT^b} \> \pc\con0 \> \g{T^c} \> \pc\con0 \> \g{H}
  \end{qcirc}
\]

By setting~$a=c=0$ and~$b=1$ one obtains a circuit that realizes the Toffoli
gate and has the same characteristics as the circuit
in \eqref{eq:amy-toffoli}, \emph{i.e.}~a $T$-depth of~$3$ with a total depth
of~$10$. The more complicated derivation of precisely the same circuit as
in~\eqref{eq:amy-toffoli} is given in the appendix.

Depending on how~$a$, $b$, and~$c$ are assigned values of 0 and 1, we can realize
four different circuits as illustrated in the following table where the first row yields a circuit for the Toffoli gate.  Note that each row represents two possible assignments, \emph{e.g.} setting~$a=c=1$ and~$b=0$ also realizes a Toffoli gate.  The two possibilities correspond to performing the required rotations in opposite direction around the Bloch sphere.
\begin{center}
\small
\begin{tabularx}{.8\linewidth}{Xl}\hline
  Assignment & Control function \\\hline
  $c=a$ and $b\neq a$ & $f(x_1,x_2)=x_1x_2$ \\
  $b=a$ and $c\neq a$ & $f(x_1,x_2)=x_1\bar x_2$ \\
  $b=c$ and $a\neq c$ & $f(x_1,x_2)=\bar x_1x_2$ \\
  $a=b=c$ & $f(x_1,x_2)=x_1\lor x_2$\\\hline
\end{tabularx}
\end{center}

The last entry in the table above can be used to realize $x_3\oplus \bar x_1\bar x_2$ by adding a NOT gate to the right of the right-hand $H$.
% Circuit realizations for those functions have been found using exhaustive
% search techniques in~\cite[Figs.~7(b)+7(c)]{AMM+:2013}, however, those
% realizations do not show the same circuit structure.

Another example is given in Fig.~\ref{fig:full-adder} in which the
full adder circuit in~\cite[Fig.~11]{AMM+:2013} is derived starting from two
Peres gates, each shown as a Toffoli gate followed by a CNOT.
After the NCV expansion has been applied two controlled $V$ gates cancel and we can
reorder the gates before mapping them to controlled $S$ gates.
After the controlled $S$ gates have been mapped according to~\eqref{eq:theo-remove-one-line},
the next steps aim to align the~$T$ gates
into as few stages as possible by moving the blocking CNOT gates
using~\eqref{eq:cnot-rule}.
In that way a $T$-depth of~$2$ is readily achieved.  However, in order to get the
same total depth as~\cite[Fig.~11]{AMM+:2013}, one needs to perform the not so obvious
step of adding two CNOT gates into the circuit (as depicted in the dashed
rectangle) followed by more CNOT moves and reductions.

It is very interesting that by applying \eqref{eq:pauli-root-to-pauli-root}
and~\eqref{eq:turn-pauli-with-hadamard} followed by the canceling of neighboring
Hadamard gates, one can easily transform the resulting $T$ gate circuit into
a~$W$ gate circuit
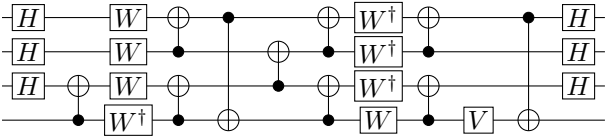
\[
  \qcircx=.95cm
  \def\qcircscale{.7}
  \begin{qcirc}
  \g{H} \> \emp \> \g{W} \> \targ \> \pc \> \emp \> \targ \> \g{W^\dagger} \> \targ \> \emp \> \pc \> \g{H} \cr
  \g{H} \> \emp \> \g{W} \> \pc\con0 \> \emp \> \targ \> \pc\con0 \> \g{W^\dagger} \> \pc\con0 \> \emp \> \emp \> \g{H} \cr
  \g{H} \> \targ \> \g{W} \> \targ \> \emp \> \pc\con1 \> \targ \> \g{W^\dagger} \> \targ \> \emp \> \emp \> \g{H} \cr
  \emp \> \pc\con2 \> \g{W^\dagger} \> \pc\con2 \> \targ\con0 \> \emp \> \pc\con2 \> \g{W} \> \pc\con2 \> \g{V} \> \targ\con0 \> \emp
  \end{qcirc}
\]
where~$W=\sqrt[4]{X}$.
This translation is applicable in general and not specific to this example.
Moreover, other Pauli roots can be used together with their respective
translation matrices.
As a result, the size of the circuit does not depend on the underlying Pauli
root quantum gate library.
However, some optimization techniques may only be possible for certain Pauli
matrices.

\section{Clifford groups}
The Clifford group is usually defined over~$C_1(X)$, $S$, and~$H$ in the
literature.
In this section, we investigate how the Pauli roots relate to Clifford
groups and can be utilized in order to express Clifford groups in a general
manner.
We first derive some useful equalities involving Pauli matrices and their roots.

The Pauli matrices form an orthogonal basis for the real Hilbert space and
therefore one Pauli matrix cannot be expressed as a linear combination of the
others.
However, the Pauli matrices obey the commutation relation $[\sigma_a,\sigma_b]
=2i\sum_{c=1}^3\varepsilon_{abc}\sigma_c$ where~$\varepsilon_{abc}=\frac{(a-b)(b-c)(c-a)}{2}$
is the Levi-Civita symbol and the anticommutation relation~$\{\sigma_a,\sigma_b\}
=2\delta_{ab}I$ where~$\delta_{ab}=\langle a|I_3|b\rangle$ is the Kronecker delta.
Exploiting these relations, the following equation is obtained
\begin{align}
  \label{eq:communtation-step}
  {}[\sigma_a,\sigma_b]+\{\sigma_a,\sigma_b\}&=(\sigma_a\sigma_b-\sigma_b\sigma_a)
  +(\sigma_a\sigma_b+\sigma_b\sigma_a)\nonumber \\
  &=2i\sum_{c=1}^3\varepsilon_{abc}\sigma_c+2\delta_{ab}I
\end{align}
leading to
\begin{equation}
  \label{eq:pauli-as-others}
  \sigma_a\sigma_b=i\sum_{c=1}^3\varepsilon_{abc}\sigma_c+\delta_{ab}I.
\end{equation}
Hence, an equation can be derived that expresses
a Pauli root~$\sigma_b$ in terms of a Pauli root~$\sigma_a$ and the square root
of~$\sigma_c$ for~$|\{a,b,c\}|=3$, \emph{i.e.}
\begin{equation}
  \label{eq:pauli-root-as-others}
  \sqrt[k]{\sigma_b}
  =\varepsilon_{abc}\left(\sqrt{\sigma_c}\sqrt[k]{\sigma_a}\sqrt{\sigma_c}^\dagger\right).
\end{equation}
Let
\begin{equation}
  \label{eq:pauli-group}
  \mathcal{P}_n=\{\sigma_{j_1}\otimes\dots\otimes\sigma_{j_n}\mid
  \{j_1,\dots,j_n\}\subseteq\{0,1,2,3\}\}
\end{equation}
be the \emph{Pauli group} where~$\sigma_0=I$, \emph{i.e.}~the set of all~$n$-fold tensor
products of the identity and the Pauli matrices~\cite{DG:1998}.
Following the definition in~\cite{AMM+:2013}, the \emph{Clifford
group}~$\mathcal{C}_n$ is the normalizer~\cite{DG:1998} of~$\mathcal{P}_n$ in
the group of all unitary~$2^n\times2^n$ matrices~$\mathrm{U}(2^n)$, \emph{i.e.}
\begin{equation}
  \label{eq:clifford-group}
  \mathcal{C}_n=\{U\in\mathrm{U}(2^n)\mid
  \forall P\in\mathcal{P}_n:UPU^\dagger\in\mathcal{P}_n\}.
\end{equation}

The Clifford group can be generated by~$C_1(X)$,
$S=\sqrt{Z}$, and~$H$~\cite{DG:1998}.
Other matrices can then be derived \emph{e.g.}
\begin{equation}
  \label{eq:clifford-examples}
  Z\!=\!SS, \, X\!=\!HSSH, \,
  Y\!\stackrel{\text{(\ref{eq:pauli-root-as-others})}}{=}\!SXS^\dagger\!=\!SHSSHSSS.
\end{equation}
Other choices are possible.  For example, the Clifford group can be generated by
\emph{e.g.}~$C_1(X)$, $V$, and~$H$, since $S=HVH$.
\begin{theorem}
\label{theo:clifford}
The gates~$C_1(\sigma_a)$, $\sqrt{\sigma_a}$, and~$\rho_{ab}$ generate the
Clifford group where~$a\neq b$.
\end{theorem}

\begin{proof}
According to \eqref{eq:pauli-root-to-pauli-root} one can obtain a second
Pauli square root~$\sqrt{\sigma_b}$ from~$\sqrt{\sigma_a}$ by multiplying it
with~$\rho_{ab}$ on both sides.
The third square root can then be obtained from~\eqref{eq:pauli-root-as-others}.
\end{proof}
\begin{corollary}
The gates~$C_1(\sigma_a)$, $\sqrt{\sigma_b}$, and~$\rho_{ab}$ generate the
Clifford group.
\end{corollary}

Based on the results of this section, circuits built using gates from the
Clifford+$T$ library can also be expressed using a more general Clifford group
together with a fourth root of a corresponding Pauli matrix.

\section{Negator operations}
In~\cite{VB:2013} the authors introduce a related gate called the
\texttt{NEGATOR}~$N(\theta)$ which they obtain by replacing~$\frac{\pi}{2k}$ by
an angle~$\theta$ in the matrix representation for the root of Pauli~$X$, \emph{i.e.}
\begin{equation}
  \label{eq:negator}
  N(\theta)=I+i\sin\theta+e^{i\theta}\left(I-X\right).
\end{equation}
This gate can be generalized for all Pauli matrices as
\begin{equation}
  \label{eq:negator-generlized}
  N_a(\theta)=I+i\sin\theta+e^{i\theta}\left(I-\sigma_a\right)
\end{equation}
and it can be proven that the matrix is obtained from the corresponding Pauli
root by replacing~$\frac{\pi}{2k}$ with~$\theta$.

\begin{theorem}
  \label{theo:negator}
  It holds, that
  \begin{equation}
    \label{eq:theo-negator}
    N_a\left(\tfrac{\theta}{2}\right)=e^{\frac{i\theta}{2}}R_a(\theta).
  \end{equation}
Note that~$\root k\of{\sigma_a}=
e^{\frac{i\pi}{2k}}R_a\left(\tfrac{\pi}{k}\right)$.
\end{theorem}

\begin{proof}
We rewrite \eqref{eq:theo-negator} to~$R_a(\theta)=e^{-\frac{i\theta}{2}}
N_a\left(\tfrac{\theta}{2}\right)$ and after expanding~$N_a(\tfrac{\theta}{2})$ we have
\begin{align*}
  &e^{-\frac{i\theta}{2}}\left(I+i\sin{\tfrac{\theta}{2}}
  e^{\frac{i\theta}{2}}\left(I-\sigma_a\right)\right)\\
  =\;&e^{-\frac{i\theta}{2}}I+i\sin{\tfrac{\theta}{2}}I
  -i\sin{\tfrac{\theta}{2}}\sigma_a
  \!=\!\cos{\tfrac{\theta}{2}}I
  \!-\!i\sin{\tfrac{\theta}{2}}\sigma_a
  \!=\!R_a(\theta)
\end{align*}
which concludes the proof.
\end{proof}

\begin{corollary}
  From Theorem~\ref{theo:clifford} it can be seen that~$C_1(\sigma_a)$, $N_a(\theta)$,
and~$\rho_{ab}$ generate the Clifford group.
\end{corollary}

\section{Concluding Remarks}
This paper has examined relationships between the Pauli matrices and their roots with an emphasis on
circuit synthesis and optimization.  We have presented a number of useful identities and techniques.
We have also shown that by applying those methods it is possible to systematically derive circuits employing Clifford+$T$ gates from NCV circuits
without using expensive exhaustive search techniques.
Our ongoing work is to incorporate the techniques discussed in this paper into a synthesis algorithm opening the
way to the synthesis and optimization of larger Clifford+$T$ gate circuits.

\bibliography{library}

\newpage
\appendix
\section{Deriving the circuit in~(\ref{eq:amy-toffoli})}
\begingroup
  \qcircx=.8cm
  \def\qcircscale{.7}
Before we derive the circuit, we show one further identity:
\[
  \qcircx=.73cm
  \begin{qcirc}
    \pc \> \emp \> \pc \cr
    \targ\con0 \> \g{T} \> \targ\con0
  \end{qcirc}
  \!=\!
  \begin{qcirc}
    \pc \> \g{T^\dagger} \> \pc \> \g{T} \> \emp \cr
    \targ\con0 \> \g{T} \> \targ\con0 \> \g{T^\dagger} \> \g{T}
  \end{qcirc}
  \!\stackrel{\eqref{eq:theo-remove-one-line}}{=}\!
  \begin{qcirc} \pc \> \g{T} \cr \g{S}\con0 \> \g{T} \end{qcirc}
  \!\stackrel{\eqref{eq:switch-z-roots}}{=}\!
  \begin{qcirc} \g{S} \> \g{T} \cr \pc\con0 \> \g{T} \end{qcirc}
\]
\begin{equation}
  \qcircx=.73cm
  \label{eq:swap-t}
  \stackrel{\eqref{eq:theo-remove-one-line}}{=}
  \begin{qcirc}
    \targ \> \g{T} \> \targ \> \g{T^\dagger} \> \g{T} \cr
    \pc\con0 \> \g{T^\dagger} \> \pc\con0 \> \g{T} \> \emp
  \end{qcirc}
  =
  \begin{qcirc}
    \targ \> \g{T} \> \targ \cr
    \pc\con0 \> \emp \> \pc\con0
  \end{qcirc}
\end{equation}

We start as before by mapping the controlled $V$ and $V^\dagger$ gates to $S$ and $S^\dagger$ gates as follows.
\[
\begin{qcirc} \pc \cr \pc \cr \targ\con0 \end{qcirc}
=
\begin{qcirc}
  \emp       \> \pc        \> \emp       \> \pc              \> \pc  \cr
  \pc        \> \targ\con0 \> \pc              \> \targ\con0 \> \emp \cr
  \g{V}\con1 \> \emp       \> \g{V^\dagger}\con1 \> \emp       \> \g{V}\con0
\end{qcirc}
=
\begin{qcirc}
  \emp  \> \emp       \> \pc        \> \emp       \> \pc              \> \pc        \> \emp  \cr
  \emp  \> \pc        \> \targ\con0 \> \g{S^\dagger} \> \targ\con0 \> \emp       \> \emp \cr
  \g{H} \> \g{S}\con1 \> \emp       \> \pc\con1 \> \emp       \> \g{S}\con0 \> \g{H}
\end{qcirc}
\]
Replacing the controlled~$S$ and $S^\dagger$ gates by  cascades consisting of uncontrolled
$T$ and $T^\dagger$ gates yields
\[
\begin{qcirc}
  \emp \> \emp \> \emp \> \emp \> \emp \> \pc \> \emp \> \emp \> \emp \> \emp \> \pc \> \pc \> \g{T} \> \pc \> \emp \cr
  \emp \> \emp \> \pc \> \g{T} \> \pc \> \targ\con0 \> \g{T^\dagger} \> \targ \> \g{T} \> \targ \>  \targ\con0 \> \emp \> \emp \> \emp \> \emp \cr
  \g{H} \> \g{T} \> \targ\con1 \> \g{T^\dagger} \> \targ\con1 \> \emp \> \emp \> \pc\con1 \> \g{T^\dagger} \> \pc\con1 \> \g{T} \> \targ\con0 \> \g{T^\dagger} \> \targ\con0 \> \g{H}
  \begin{scope}[<->,thick,gray!30!black,every node/.style={midway,below,inner sep=1pt,font=\footnotesize}]
    \draw[->] ([yshift=7pt] qc012) to[bend right=5] ([yshift=7pt] qc01);
    \draw[->] ([yshift=7pt] qc13) to[bend right=5] ([yshift=7pt] qc11);
    \draw ([yshift=-9pt] qc28) to[bend right=10] node {cancel} ([yshift=-9pt] qc210);
  \end{scope}
\end{qcirc}
\]
After moving~$T$ gates on the first and second line towards the left and
canceling the~$T$ with the~$T^\dagger$ gate on the third line we have
\[
\begin{qcirc}
  \emp \> \g{T} \> \emp \> \emp \> \emp \> \pc \> \emp \> \emp \> \emp \> \emp \> \pc \> \pc \> \emp \> \pc \> \emp \cr
  \emp \> \g{T} \> \pc \> \emp \> \pc \> \targ\con0 \> \g{T^\dagger} \> \targ \> \g{T} \> \targ \>  \targ\con0 \> \emp \> \emp \> \emp \> \emp \cr
  \g{H} \> \g{T} \> \targ\con1 \> \g{T^\dagger} \> \targ\con1 \> \emp \> \emp \> \pc\con1 \> \emp \> \pc\con1 \> \emp \> \targ\con0 \> \g{T^\dagger} \> \targ\con0 \> \g{H}
  \begin{scope}[<->,thick,gray!30!black,every node/.style={midway,below,inner sep=1pt,font=\footnotesize}]
    \draw[->] ([yshift=-7pt] qc211) to[bend left=10] ([yshift=-7pt] qc27);
  \end{scope}
\end{qcirc}
\]
In this circuit we move a CNOT to the left by applying~\eqref{eq:cnot-rule}
resulting in
\[
\begin{qcirc}
  \emp \> \g{T} \> \emp \> \emp \> \emp \> \pc \> \emp \> \pc \> \pc \> \emp \> \emp \> \emp \> \emp \> \pc \> \emp \cr
  \emp \> \g{T} \> \pc \> \emp \> \pc \> \targ\con0 \> \g{T^\dagger} \> \targ\con0 \> \emp \> \targ \> \g{T} \> \targ \> \emp \> \emp \> \emp \cr
  \g{H} \> \g{T} \> \targ\con1 \> \g{T^\dagger} \> \targ\con1 \> \emp \> \emp \> \emp \> \targ\con0 \> \pc\con1 \> \emp \> \pc\con1 \> \g{T^\dagger} \> \targ\con0 \> \g{H}
  \begin{scope}[<->,thick,gray!30!black,every node/.style={midway,below,inner sep=1pt,font=\footnotesize}]
    \draw[->] ([yshift=5pt] qc05) to[bend right=10] ([yshift=5pt] qc02);
    \draw[->] ([yshift=9pt] qc16) to[bend right=10] ([yshift=9pt] qc13);
    \draw[->] ([yshift=-9pt] qc212) to[bend left=10] ([yshift=-9pt] qc210);
  \end{scope}
\end{qcirc}
\]
After moving the selected gates to the left the circuit is as follows:
\[
\begin{qcirc}
  \emp \> \g{T} \> \pc \> \emp \> \emp \> \emp \> \pc \> \pc \> \emp \> \emp \> \emp \> \pc \> \emp \cr
  \emp \> \g{T} \> \targ\con0 \> \pc \> \g{T^\dagger} \> \pc \> \targ\con0 \> \emp \> \targ \> \g{T} \> \targ \> \emp \> \emp \cr
  \g{H} \> \g{T} \> \emp \> \targ\con1 \> \g{T^\dagger} \> \targ\con1 \> \emp \> \targ\con0 \> \pc\con1 \> \g{T^\dagger} \> \pc\con1 \> \targ\con0 \> \g{H}
  \draw[gray!30!black,thick,shorten <=-7pt,shorten >=-7pt] ([yshift=-11pt] qc22) -- node[below,inner sep=1.5pt,font=\footnotesize] {swap} ([yshift=-11pt] qc26);
\end{qcirc}
\]
In this step we make use of the identity in~\eqref{eq:swap-t} which changes the
position and orientation of the selected gates:
\[
\begin{qcirc}
  \emp \> \g{T} \> \targ \> \emp \> \g{T^\dagger} \> \emp \> \targ \> \pc \> \emp \> \emp \> \emp \> \pc \> \emp \cr
  \emp \> \g{T} \> \pc\con0 \> \targ \> \g{T^\dagger} \> \targ \> \pc\con0 \> \emp \> \targ \> \g{T} \> \targ \> \emp \> \emp \cr
  \g{H} \> \g{T} \> \emp \> \pc\con1 \> \emp \> \pc\con1 \> \emp \> \targ\con0 \> \pc\con1 \> \g{T^\dagger} \> \pc\con1 \> \targ\con0 \> \g{H}
  \begin{scope}[<->,thick,gray!30!black,every node/.style={midway,below,inner sep=1pt,font=\footnotesize}]
    \draw[->] ([yshift=-9pt] qc27) to[bend left=10] ([yshift=-9pt] qc25);
  \end{scope}
\end{qcirc}
\]
After again moving a CNOT to the left we have
\[
\begin{qcirc}
  \emp \> \g{T} \> \targ \> \emp \> \g{T^\dagger} \> \pc \> \pc \> \emp \> \emp \> \targ \> \emp \> \emp \> \emp \> \pc \> \emp \cr
  \emp \> \g{T} \> \pc\con0 \> \targ \> \g{T^\dagger} \> \emp \> \targ\con0 \> \targ \> \pc \> \pc\con0 \> \targ \> \g{T} \> \targ \> \emp \> \emp \cr
  \g{H} \> \g{T} \> \emp \> \pc\con1 \> \emp \> \targ\con0 \> \emp \> \pc\con1 \> \targ\con1 \> \emp \> \pc\con1 \> \g{T^\dagger} \> \pc\con1 \> \targ\con0 \> \g{H}
  \begin{scope}[<->,thick,gray!30!black,every node/.style={midway,below,inner sep=1pt,font=\footnotesize}]
    \draw[->] ([yshift=9pt] qc09) to[bend left=10] ([yshift=9pt] qc014);
    \draw[->] ([yshift=9pt] qc04) to[bend left=10] ([yshift=9pt] qc011);
    \draw[white,->] ([yshift=-13pt] qc25) to[bend left=10] ([yshift=-13pt] qc24);
    \draw[->] ([yshift=-9pt] qc25) to[bend left=10] ([yshift=-9pt] qc24);
  \end{scope}
\end{qcirc}
\]
Moving the CNOT to the end of the circuit allows the~$T^\dagger$ gate on the top line to  move into the third $T$-stage.
It also leads to a sequence of three CNOTs that corresponds to swapping two lines.
The two CNOTs right of the third $T$-stage can be extended to a line swapping
sequence by adding two identical CNOTs.
\[
\begin{qcirc}
  \emp \> \g{T} \> \targ \> \emp \> \pc \> \emp \> \pc \> \emp \> \emp \> \emp \> \g{T^\dagger} \> \emp \> \emp \> \pc \> \targ \cr
  \emp \> \g{T} \> \pc\con0 \> \targ \> \emp \> \g{T^\dagger} \> \targ\con0 \> \targ \> \pc \> \targ \> \g{T} \> \targ \> \pc \> \emp \> \pc\con0 \cr
  \g{H} \> \g{T} \> \emp \> \pc\con1 \> \targ\con0 \> \emp \> \emp \> \pc\con1 \> \targ\con1 \> \pc\con1 \> \g{T^\dagger} \> \pc\con1 \> \targ\con1 \> \targ\con0 \> \g{H}
\end{qcirc}
\]
Consequently, we can replace the gates by uncontrolled Fredkin gates
\[
\begin{qcirc}
  \emp \> \g{T} \> \targ \> \emp \> \pc \> \emp \> \pc \> \emp \> \g{T^\dagger} \> \emp \> \emp \> \pc \> \targ \cr
  \emp \> \g{T} \> \pc\con0 \> \targ \> \emp \> \g{T^\dagger} \> \targ\con0 \> \emp \> \g{T} \> \emp \> \targ \> \emp \> \pc\con0 \cr
  \g{H} \> \g{T} \> \emp \> \pc\con1 \> \targ\con0 \> \emp \> \emp \> \emp \> \g{T^\dagger} \> \emp \> \pc\con1 \> \targ\con0 \> \g{H}
  \swap{7}{1}{2}
  \swap{9}{1}{2}
\end{qcirc}
\]
and then apply the line swapping locally which results in the circuit from~\eqref{eq:amy-toffoli}:
\[
\begin{qcirc}
  \emp \> \g{T} \> \targ \> \emp \> \pc \> \emp \> \pc \> \g{T^\dagger} \> \emp \> \pc \> \targ \cr
  \emp \> \g{T} \> \pc\con0 \> \targ \> \emp \> \g{T^\dagger} \> \targ\con0 \> \g{T^\dagger} \> \targ \> \emp \> \pc\con0 \cr
  \g{H} \> \g{T} \> \emp \> \pc\con1 \> \targ\con0 \> \emp \> \emp \> \g{T} \> \pc\con1 \> \targ\con0 \> \g{H}
\end{qcirc}
\]
\endgroup

\end{document}